\documentclass[twocolumn]{article}%
\usepackage{amsmath}
\usepackage{amsfonts}
\usepackage{amssymb}
\usepackage{graphicx}%
\providecommand{\U}[1]{\protect\rule{.1in}{.1in}}
\newtheorem{theorem}{Theorem}

\newtheorem{lemma}[theorem]{Lemma}

\newenvironment{proof}[1][Proof]{\noindent\textbf{#1.} }{\ \rule{0.5em}{0.5em}}
\begin{document}

\title{Reversibility of distance mesures of states with some focus on total variation distance}
\author{Keiji Matsumoto\\Quantum Computation Group, National Institute of Informatics, \ \\2-1-2 Hitotsubashi, Chiyoda-ku, Tokyo 101-8430, \\e-mail:keiji@nii.ac.jp}
\maketitle

\section*{Abstract}

Consider a classical system, which is in the state described by probability
distribution $p$ or $q$, and embed these classical informations into quantum
system by a physical map $\Gamma$, $\rho=\Gamma(p)$ and $\sigma=\Gamma(q)$.
Intuitively, the pair $\{p_{\rho}^{M},p_{\sigma}^{M}\}$ of the distributions
of the data of the measurement $M$  on the pair $\{\rho,\sigma\}$ should
contain strictly less information than the pair $\{p,q\}$ provided the pair
$\{\rho,\sigma\}$ is non-commutative. Indeed, this statement had been shown if
the information is measured by $f$-divergence such that $f$ is operator
convex. In the paper, the statement is extended to the case where $f$ is
strictly convex. Also, we disprove the assertion for the total variation
distance $\Vert p-q\Vert_{1}$, the $f$-divergence with $f(r)=|1-r|$: if
$\{\rho,\sigma\}$ satisfies some not very restrictive conditions, $\Vert
p_{\rho}^{M}-p_{\sigma}^{M}\Vert_{1}$ equals $\Vert p-q\Vert_{1}$. Here we
present sufficient condition for general case, and necessary and sufficient
condition for qubit states.

\section{Introduction}

Consider a classical system, which is in the state described by probability
distribution $p$ or $q$, and embed these classical informations into quantum
system by a physical map $\Gamma$, $\rho=\Gamma(p)$ and $\sigma=\Gamma(q)$.
Intuitively, the pair $\{p_{\rho}^{M},p_{\sigma}^{M}\}$ of the distributions
of the data of the measurement $M$ on the pair $\{\rho,\sigma\}$ should
contain strictly less information than the pair $\{p,q\}$ provided the pair
$\{\rho,\sigma\}$ is non-commutative. Indeed, this statement had been shown if
the information is measured by $f$-divergence such that $f$ is operator convex
\cite{HiaiMilan}. In the paper, the statement is extended to the case where
$f$ is strictly convex. Also, we disprove the assertion for the total
variation distance $\Vert p-q\Vert_{1}$, the $f$-divergence with $f(r)=|1-r|$:
if $\{\rho,\sigma\}$ satisfies some not very restrictive conditions, $\Vert
p_{\rho}^{M}-p_{\sigma}^{M}\Vert_{1}$ equals $\Vert p-q\Vert_{1}$. Here we
present sufficient condition for general case, and necessary and sufficient
condition for qubit states.

\section{Embedding Classical Information Into Quantum States}

Consider a classical memory system, whose state is described by probability
distribution $p$ or $q$ depending on the value of the bit recorded. Suppose we
embed this information into quantum system by some physical operation $\Gamma
$, or completely positive trace preserving (CPTP) map from commutative system
into operators $\mathcal{B}\left(  \mathcal{H}\right)  $ over Hilbert space
$\mathcal{H}$. (In this paper, we stick to the finite dimensional case.) Then
we obtain a quantum system whose state is either $\rho=\Gamma(p)$ or
$\sigma=\Gamma(q)$ depending on the value of the bit.

Suppose now we are given $\left\{  \rho,\sigma\right\}  $, and the question is
how much of information is contained in $\{p,q\}$. The answer relies on the
measure of information, and also on the choice of $\Gamma$. In the paper, we
use $f$- divergence between $p$ and $q$ to measure the amount of information:%
\[
\mathrm{D}_{f}(p\Vert q)\colon=\sum_{x\in\mathcal{X}}q_{x}f(p_{x}/q_{x}),
\]
where $\mathcal{X}$ is a finite set, and $f$ is a convex function. In the
definition, we used the convention
\begin{align*}
0\cdot f\left(  p/0\right)  \colon &  =p\hat{f}(0),\\
\text{ }\hat{f}(0)\colon &  =\lim_{\varepsilon\downarrow0}\varepsilon
f(1/\varepsilon)=\lim_{r\rightarrow\infty}r^{-1}f(r).
\end{align*}
By choosing $f$ properly, $f$- divergence represents almost all frequently
used distance measures (or their monotone function): relative entropy
(Kullback-Leibler divergence), Renyi relative entropy, total variation
distance, and so on.

As for dependence of $\Gamma$, we suppose the encoder did their best: thus our
question is to find
\[
\mathrm{D}_{f}^{\max}(\rho\Vert\sigma)=\min_{\left(  \Gamma,\{p,q\}\right)
}\mathrm{D}_{f}(p\Vert q),
\]
where $\left(  \Gamma,\{p,q\}\right)  $ moves over all the triple satisfying
\begin{equation}
\rho=\Gamma(p),\,\sigma=\Gamma(q). \label{r-test}%
\end{equation}
Since resulted from optimization problem, $\mathrm{D}_{f}^{\max}$ is monotone
decreasing by CPTP maps. Also, when $\left[  \rho,\sigma\right]  =0$, it
reduces to its classical version $\mathrm{D}_{f}$.

If $f$ is operator convex,
\[
\mathrm{D}_{f}^{\max}(\rho\Vert\sigma)=\mathrm{tr}\,\sigma f(\sigma^{-1/2}%
\rho\sigma^{-1/2}),
\]
provided $\sigma>0$ and $\rho>0$ \cite{Matsumoto}. Examples are $r\ln r$, and
\[
f_{\alpha}(r)=(\pm)r^{\alpha},\,\,\alpha\in(-1,1)\text{,}%
\]
where the sign $\pm$ is chosen so that $f_{\alpha}$ is convex are operator
convex. The former and the latter corresponds to relative entropy and relative
Renyi entropy, respectively.

However, the function $|1-r|$, which corresponds to total variation distance
\[
\left\Vert p-q\right\Vert _{1}=\sum_{x\in\mathcal{X}}\left\vert p_{x}%
-q_{x}\right\vert ,
\]
is \textit{not} operator convex, and there is no known closed formula for
$\mathrm{D}_{\left\vert 1-r\right\vert }^{\max}(\rho\Vert\sigma)$.

\section{Reversibility}

To read classical information from a quantum source $\left\{  \rho
,\sigma\right\}  $, a measurement $M$ is applied to the system, to produce the
probability distributions:
\[
p_{x}^{\prime}=\mathrm{tr}\,M_{x}\rho,q_{x}^{\prime}=\mathrm{tr}\,\sigma
M_{x}.
\]
Obviously,
\begin{equation}
\mathrm{D}_{f}^{\max}(\rho\Vert\sigma)\geq\mathrm{D}_{f}(p^{\prime}\Vert
q^{\prime}). \label{measure-decrease}%
\end{equation}
If $\left[  \rho,\sigma\right]  =0$, the identity in the above inequality
holds: in fact, if $f$ is strictly convex
\[
f(cr_{1}+(1-c)r_{2})<cf(r_{1})+(1-c)f(r_{2}),\,\,\,\,0<c<1,
\]
this is the only possible case for the equality to holds. Some preparations
are necessary to prove the assertion. Let $\left\{  p,q\right\}  $ and
$\left\{  p^{\prime},q^{\prime}\right\}  $ be a probability distribution over
a finite set $\mathcal{X}$ and $\mathcal{Y}$, respectively. Define
$\mathcal{X}_{0}\colon=\{x;q_{x}=0\}$, and
\[
r_{x}\colon=\left\{
\begin{array}
[c]{cc}%
p_{x}/q_{x}, & x\notin\mathcal{X}_{0},\\
\infty, & x\in\mathcal{X}_{0}.
\end{array}
\right.
\]
$\mathcal{Y}_{0}$ and $r_{y}^{\prime}$ are defined almost analogously.

\begin{lemma}
\label{lem:r-block}Suppose there is a transition probability $P(y|x)$ with%
\[
p_{y}^{\prime}=\sum_{x\in\mathcal{X}}P\left(  y|x\right)  p_{x},\,\,q_{y}%
^{\prime}=\sum_{x\in\mathcal{X}}P\left(  y|x\right)  q_{x},
\]
and there is a strictly convex function on $(0,\infty)$ with
\[
\mathrm{D}_{f}(p\Vert q)=\mathrm{D}_{f}(p^{\prime}\Vert q^{\prime})<\infty.
\]
Then $P(x|y)=0$ for all $x$ and $y$ with $r_{x}\neq r_{y}$. \ 
\end{lemma}

\begin{proof}
First, we prove the case where $\hat{f}(0)<\infty$. Then $f(r)$ decomposes
into
\[
f(r)=f_{0}(r)+\hat{f}(0)r,
\]
where $f_{0}$ is monotone non\thinspace-\thinspace increasing. Then
\[
\mathrm{D}_{f}(p\Vert q)=\sum_{x\notin\mathcal{X}_{0}}q_{x}f_{0}(r_{x}%
)+\sum_{x\in\mathcal{X}}p_{x}\hat{f}(0).
\]%
\begin{align*}
&  \mathrm{D}_{f}(p^{\prime}\Vert q^{\prime})\\
&  =\sum_{y\notin\mathcal{Y}_{0}}q_{y}^{\prime}f_{0}(r_{y}^{\prime}%
)+\sum_{y\in\mathcal{Y}}p_{y}\hat{f}(0)\\
&  =\sum_{y\notin\mathcal{Y}_{0}}q_{y}^{\prime}f_{0}(\sum_{x\notin
\mathcal{X}_{0}}Q(x|y)r_{x}+\frac{1}{q_{y}^{\prime}}\sum_{x\in\mathcal{X}_{0}%
}P(y|x)p_{x})\\
&  +\sum_{y\in\mathcal{Y}}p_{y}\hat{f}(0),
\end{align*}
where
\[
Q(x|y):=\frac{q_{x}}{q_{y}^{\prime}}P(y|x),\,x\notin\mathcal{X}_{0}%
,\,y\notin\mathcal{Y}_{0}.
\]
Since$\frac{1}{q_{y}^{\prime}}\sum_{x\in\mathcal{X}_{0}}P(y|x)q_{x}\geq0$ and
$\sum_{x\notin\mathcal{X}_{0}}Q(x|y)=1$,
\begin{align*}
\mathrm{D}_{f}(p^{\prime}\Vert q^{\prime})  &  \leq\sum_{y\notin
\mathcal{Y}_{0}}q_{y}^{\prime}f_{0}(\sum_{x\notin\mathcal{X}_{0}}%
Q(x|y)r_{x})+\sum_{y\in\mathcal{Y}}p_{y}\hat{f}(0)\\
&  \leq\sum_{x\notin\mathcal{X}_{0}}\sum_{y\notin\mathcal{Y}_{0}}%
Q(x|y)q_{y}^{\prime}f_{0}(r_{x})+\sum_{y\in\mathcal{Y}}p_{y}\hat{f}(0)\\
&  =\sum_{x\notin\mathcal{X}_{0}}q_{x}f_{0}(r_{x})+\sum_{y\in\mathcal{Y}}%
p_{y}\hat{f}(0)\\
&  =\mathrm{D}_{f}(p\Vert q).
\end{align*}
Since $f$ is strictly convex, $\mathrm{D}_{f}(p^{\prime}\Vert q^{\prime
})=\mathrm{D}_{f}(p\Vert q)$ holds only if
\[
\sum_{x\in\mathcal{X}_{0}}P(y|x)p_{x}=0,\,y\notin\mathcal{Y}_{0}.
\]
and
\[
Q(x|y)=0,\,\,r_{x}\neq r_{y}^{\prime},\,x\notin\mathcal{X}_{0},\,y\notin
\mathcal{Y}_{0}.\,
\]
These are equivalent to
\[
P(y|x)=0,r_{x}\neq r_{y}^{\prime},\,y\notin\mathcal{Y}_{0}.
\]
Also, the condition $q_{y}^{\prime}=\sum_{x\in\mathcal{X}}P\left(  y|x\right)
q_{x}$ implies
\[
P(y|x)=0,\,y\in\mathcal{Y}_{0},\,x\notin\mathcal{X}_{0}.
\]
Therefore, we have the assertion provided $\hat{f}(0)<\infty$.

Next, we study the case where $\hat{f}(0)=\infty$. Then $\mathrm{D}_{f}(p\Vert
q)=\mathrm{D}_{f}(p^{\prime}\Vert q^{\prime})<\infty$ implies $\mathcal{X}%
_{0}=\emptyset$ and $\mathcal{Y}_{0}=\emptyset$. Then doing almost analogously
as above, we have the assertion.
\end{proof}

\begin{theorem}
Suppose $f$ is strictly convex function on $(0,\infty)$, and $\mathrm{D}%
_{f}^{\max}(\rho\Vert\sigma)<\infty$. Then the equality in the inequality
(\ref{measure-decrease}) holds only if $[\rho,\sigma]=0.$
\end{theorem}

If $f$ is non\thinspace-\thinspace linear and operator convex, it is strictly
convex. Therefore, the theorem applies to relative and Renyi relative entropy.

\begin{proof}
Let $\left(  \Gamma,\left\{  p,q\right\}  \right)  $ be a triplet achieving
$\mathrm{D}_{f}^{\max}(\rho\Vert\sigma)=\mathrm{D}_{f}(p\Vert q)$. Then the
equality in the inequality (\ref{measure-decrease}) holds only if
\[
\mathrm{D}_{f}(p\Vert q)=\mathrm{D}_{f}(p^{\prime}\Vert q^{\prime})<\infty.
\]
Since the composition of $\Gamma$ followed by the measurement $M$ is a linear,
positive, and probability preserving map, there is a transition probability
$P(y|x)$ such that
\[
p_{y}^{\prime}=\sum_{x\in\mathcal{X}}P(y|x)\,p_{x},q_{y}^{\prime}=\sum
_{x\in\mathcal{X}}P(y|x)q_{x},
\]
and
\[
P(y|x)=\mathrm{tr}\,M_{y}\Gamma(\delta_{x}),
\]
where $\delta_{x}$ is delta distribution at $x$. Therefore, by Lemma\thinspace
\ref{lem:r-block}, \thinspace$\mathrm{tr}\,M_{y}\Gamma(\delta_{x})=0$ provided
$r_{x}\neq r_{y}^{\prime}$.

Define
\begin{align*}
\rho_{r}\colon &  =\sum_{x:r_{x}=r}p_{x}\Gamma(\delta_{x}),\,\sigma_{r}%
\colon=\sum_{x:r_{x}=r}q_{x}\Gamma(\delta_{x}),\,\\
\tilde{M}_{r}\colon &  =\sum_{y:r_{y}^{\prime}=r}M_{y}.
\end{align*}
Then if $r<\infty$, observe $\rho_{r}=r\sigma_{r}$, and
\[
\mathrm{tr}\,\rho_{r}\tilde{M}_{r^{\prime}}=\mathrm{tr}\,\sigma_{r}\tilde
{M}_{r^{\prime}}=0,\,r\neq r^{\prime}.
\]
Therefore, supports of positive operators
\[
\left\{  \sigma_{r},r\in\lbrack0,\infty)\right\}  \cup\left\{  \rho_{\infty
}\right\}
\]
are non\thinspace-\thinspace overlapping with each other.
\end{proof}

Therefore, the assertion $\left[  \rho,\sigma\right]  =0$ follows since
\[
\rho=\sum_{r\in\lbrack0,\infty)}r\sigma_{r}+\rho_{\infty},\,\sigma=\sum
_{r\in\lbrack0,\infty)}\sigma_{r}.
\]

This theorem means that the classical information embedded into non-orthogonal
states cannot be recovered completely by any measurement. At first glance, the
statement seems almost trivial, but in the proof we fully exploit the fact
that $f$ is strictly convex, and in fact, is not true if the information
measure is total variation distance.

\section{Total variation distance}

\subsection{Set up and a general formula}

Total variation distance, or the divergence corresponding to $f\left(
r\right)  =\left\vert 1-r\right\vert $, is one of most frequently used
distance measures between two probability distributions. Its most common
quantum version is
\[
\left\Vert \rho-\sigma\right\Vert _{1}=\sup_{M}\left\Vert P_{\rho}%
^{M}-P_{\sigma}^{M}\right\Vert _{1},
\]
where $P_{\rho}^{M}$ is the distribution of the outcome of the measurement $M$
under $\rho$. Obviously,
\[
\mathrm{D}_{\left\vert 1-r\right\vert }^{\max}\left(  \rho\Vert\sigma\right)
\geq\left\Vert \rho-\sigma\right\Vert _{1}.
\]
Given a triple $\left(  \Gamma,\left\{  p,q\right\}  \right)  $ of $\left\{
\rho,\sigma\right\}  $, we define $\left(  \Gamma^{\prime},\left\{  p^{\prime
},q^{\prime}\right\}  \right)  $, where $\left\{  p^{\prime},q^{\prime
}\right\}  $ are probability distributions on $\left\{  0,1,2\right\}  $:
\begin{align}
\Gamma^{\prime}\left(  \delta_{0}\right)   &  :=\frac{1}{\mathrm{tr}%
\,A}A,\Gamma^{\prime}\left(  \delta_{1}\right)  :=\frac{\rho-A}{\mathrm{tr}%
\,\left(  \rho-A\right)  },\nonumber\\
\Gamma^{\prime}\left(  \delta_{2}\right)   &  :=\frac{\sigma-A}{\mathrm{tr}%
\,\left(  \sigma-A\right)  },\nonumber\\
p^{\prime}\left(  0\right)   &  :=\mathrm{tr}\,A,\,\,p^{\prime}\left(
1\right)  :=\mathrm{tr}\,\left(  \rho-A\right)  ,\,\,\,p^{\prime}\left(
2\right)  :=0,\nonumber\\
q^{\prime}\left(  0\right)   &  :=\mathrm{tr}\,A,\,\,q^{\prime}\left(
1\right)  :=0,\,\,\,q^{\prime}\left(  2\right)  :=\mathrm{tr}\,\left(
\sigma-A\right)  . \label{opt-|1-l|}%
\end{align}
where
\[
A:=\sum_{x\in\mathcal{X}}\min\left\{  p\left(  x\right)  ,q\left(  x\right)
\right\}  \Gamma\left(  \delta_{x}\right)  .
\]
Then $\left(  \Gamma^{\prime},\left\{  p^{\prime},q^{\prime}\right\}  \right)
$ satisfies (\ref{r-test}) and $\left\Vert p^{\prime}-q^{\prime}\right\Vert
_{1}=\left\Vert p-q\right\Vert _{1}$.

(Intuitively, $\Gamma^{\prime}\left(  \delta_{0}\right)  $ takes care of the
common part of two states, and $\Gamma^{\prime}\left(  \delta_{1}\right)  $
and $\Gamma^{\prime}\left(  \delta_{2}\right)  $ compensates the reminder.)

Therefore, without loss of generality, we may restrict ourselves to the one in
the form of (\ref{opt-|1-l|}), where $A$ is an operator with
\[
A\geq0,\rho\geq A,\sigma\geq A.
\]
Therefore, we have:
\begin{align}
&  \mathrm{D}_{\left\vert 1-r\right\vert }^{\max}\left(  \rho\Vert
\sigma\right) \nonumber\\
&  =\inf\left\{  \mathrm{tr}\,\left(  \rho+\sigma-2A\right)  ;A\geq0,\rho\geq
A,\sigma\geq A\right\}  . \label{Dmax-|1-l|}%
\end{align}

\subsection{Reversibility}

In this subsection and the next, suppose $\mathrm{tr}\,\rho=\mathrm{tr}%
\,\sigma=1$. We study the conditions for
\begin{equation}
\mathrm{D}_{\left\vert 1-r\right\vert }^{\max}\left(  \rho\Vert\sigma\right)
=\left\Vert \rho-\sigma\right\Vert _{1}. \label{D=TV}%
\end{equation}
This implies that any quantum version of statistical distance $\mathrm{D}%
_{\left\vert 1-r\right\vert }^{Q}\left(  \rho\Vert\sigma\right)  $ equals to
$\left\Vert \rho-\sigma\right\Vert _{1}$. Intuitively, this means classical
statistical distance encoded into quantum states can be completely retrieved.
As stated, such complete retrieval of $f$\thinspace-\thinspace divergence
scarcely occurs if $f$ is operator convex and $\rho$ and $\sigma$ do not
commute. The statistical distance is very different from $f$\thinspace
-\thinspace divergence induced by an operator convex function in this respect.

If we drop the constraint $A\geq0$ and suppose $\mathrm{tr}\,\rho
=\mathrm{tr}\,\sigma$,
\begin{align*}
&  \mathrm{D}_{\left\vert 1-r\right\vert }^{\max}\left(  \rho\Vert
\sigma\right) \\
&  \geq\inf\left\{  \mathrm{tr}\,\left(  \rho+\sigma-2A\right)  ;\rho\geq
A,\sigma\geq A\right\} \\
&  =\inf\left\{  2\mathrm{tr}\,\left(  \rho-A\right)  ;\rho\geq A,\sigma\geq
A\right\} \\
&  =\inf\left\{  2\mathrm{tr}\,\left(  \rho-A\right)  ;\rho-A\geq0,\rho
-A\geq\rho-\sigma\right\} \\
&  =2\mathrm{tr}\,\left[  \rho-\sigma\right]  _{+}=\left\Vert \rho
-\sigma\right\Vert _{1}.
\end{align*}
Here, the minimum in the third line is achieved if $\rho-A=\,\left[
\rho-\sigma\right]  _{+}$. ($\left[  X\right]  _{+}$ is the positive part of
the self-adjoint operator $X$.)

Therefore, (\ref{D=TV}) holds iff \
\begin{equation}
A=\rho-\,\left[  \rho-\sigma\right]  _{+}=\frac{1}{2}\left(  \rho
+\sigma-\left\vert \rho-\sigma\right\vert \right)  \geq0. \label{A>0}%
\end{equation}
(Here, $\left\vert X\right\vert :=\sqrt{X^{\dagger}X}$.) \ Another necessary
and sufficient condition is the existence of $A$, $\Delta_{1}$, $\Delta
_{2}\geq0$ with
\begin{align}
\rho &  =A+\Delta_{1},\sigma=A+\Delta_{2},\,\label{A+d}\\
\Delta_{1}\Delta_{2}  &  =0. \label{dd=0}%
\end{align}
To see this, observe
\begin{align*}
\left\Vert \Delta_{1}-\Delta_{2}\right\Vert _{1}  &  =\left\Vert \rho
-\sigma\right\Vert _{1}\\
&  \leq\mathrm{D}_{\left\vert 1-r\right\vert }^{\max}\left(  \rho\Vert
\sigma\right) \\
&  =\min\left\{  \mathrm{tr}\,\Delta_{1}+\mathrm{tr}\,\Delta_{2}%
;(\ref{A+d}),\Delta_{1}\geq0,\Delta_{2}\geq0\right\}  .
\end{align*}
For (\ref{D=TV}) to hold, existence of $\Delta_{1}$, $\Delta_{2}$ with
$\mathrm{tr}\,\Delta_{1}+\mathrm{tr}\,\Delta_{2}=\left\Vert \Delta_{1}%
-\Delta_{2}\right\Vert _{1}$ is necessary and sufficient. Thus $\Delta
_{1}\Delta_{2}=0$.

Of course, in general, (\ref{A>0}) is not true. For example, if $\sigma
=\left\vert \psi\right\rangle \left\langle \psi\right\vert $ is a pure state,
\[
\mathrm{D}_{\left\vert 1-r\right\vert }^{\max}\left(  \rho\Vert\sigma\right)
=1-\left\langle \psi\right\vert \rho\left\vert \psi\right\rangle +\left\langle
\psi\right\vert \rho\rho_{22}^{-1}\rho\left\vert \psi\right\rangle ,
\]
where $\rho_{22}\colon=\left(  I-\left\vert \psi\right\rangle \left\langle
\psi\right\vert \right)  \rho\left(  I-\left\vert \psi\right\rangle
\left\langle \psi\right\vert \right)  $ and $\rho_{22}^{-1}$ denotes its
generalized inverse \cite{Matsumoto}.

However, if $\rho$ and $\sigma$ are very close so that
\begin{equation}
\left\Vert \,\left\vert \rho-\sigma\right\vert \right\Vert \leq\text{minimum
eigenvalue of }\rho+\sigma, \label{close}%
\end{equation}
it is true.

Another sufficient condition is
\[
\left(  \rho-\sigma\right)  ^{2}=\left\vert \rho-\sigma\right\vert ^{2}%
\leq\left(  \rho+\sigma\right)  ^{2}.
\]
To see this is sufficient, take the square root of both sides of inequality:
then we obtain (\ref{A>0}). (Recall $\sqrt{\cdot}$ is operator monotone. This
condition is not necessary, since $r^{2}$ is not operator monotone.)
Rearranging the terms, we have
\begin{equation}
\rho\sigma+\sigma\rho\geq0. \label{ls+sl}%
\end{equation}

\subsection{2\thinspace-\thinspace dimensional case}

In this subsection, we assume $\dim\mathcal{H}=2$ and $\mathrm{tr}%
\,\rho=\mathrm{tr}\,\sigma=1$, and compute the set $\left\{  \sigma
;\text{(\ref{D=TV})}\right\}  $ for each fixed $\rho$, using the necessary and
sufficient condition given by (\ref{A+d}) and (\ref{dd=0}). As it turns out,
this set is the spheroid, with focal points $\rho$ and $\mathbf{1}-\rho$, and
touching to the surface of Bloch sphere at each end of the longest axis.

Since $\mathrm{tr}\,\rho=\mathrm{tr}\,\sigma=1$,
\[
c:=\mathrm{tr}\,\Delta_{1}=\mathrm{tr}\,\Delta_{2}=1-\mathrm{tr}\,A,
\]
and
\[
0\leq c\leq1.
\]
Let $v_{\rho}$ , $v_{\sigma}$, $u_{1}$, $u_{2}$, and $u_{A}$ be the Bloch
vector of $\rho$, $\sigma$, $\frac{1}{c}\Delta_{1}$, $\frac{1}{c}\Delta_{2}$ ,
and $\frac{1}{1-c}A$ , respectively. Also, (\ref{dd=0}) holds iff $\Delta_{1}$
and $\Delta_{2}$ are rank\thinspace-\thinspace1 and $u_{2}=-u_{1}.$Therefore,
by (\ref{A+d}),
\[
v_{\rho}=cu_{1}+\left(  1-c\right)  u_{A},\,v_{\sigma}=-cu_{1}+\left(
1-c\right)  u_{A}.
\]
Therefore,%
\[
v_{\sigma}-v_{\rho}=-2cu_{1},\,v_{\sigma}-\left(  -v_{\rho}\right)  =2\left(
1-c\right)  u_{A}.
\]
Let $\left\Vert \cdot\right\Vert $ denote the Euclid norm in $\mathbb{R}^{3}$,
and
\[
\left\Vert v_{\sigma}-v_{\rho}\right\Vert +\left\Vert v_{\sigma}-\left(
-v_{\rho}\right)  \right\Vert =2\left(  c\left\Vert u_{1}\right\Vert +\left(
1-c\right)  \left\Vert u_{A}\right\Vert \right)  \leq2.
\]

The set $\left\{  \sigma;\text{(\ref{D=TV})}\right\}  $ is fairly large. For
example, if the largest eigenvalue of $\rho$ is $\leq0.85$, this occupies more
than the half of the volume of the Bloch sphere.

If
\[
\rho=\left[
\begin{array}
[c]{cc}%
a & \overline{c}\\
c & b
\end{array}
\right]  ,\sigma=\left[
\begin{array}
[c]{cc}%
a & -\overline{c}\\
-c & b
\end{array}
\right]  ,\,\,\,\left(  a\geq b\right)
\]
the minimization problem (\ref{Dmax-|1-l|}) is solved explicitly. With
$Z:=\mathrm{diag}\left(  1,-1\right)  $, $\sigma=Z\rho Z^{\dagger}$,
$\rho=Z\sigma Z^{\dagger}$. Thus, if $A$ satisfies constrains of
(\ref{Dmax-|1-l|}), so does $\frac{1}{2}\left(  ZAZ^{\dagger}+A\right)  $, and
$\mathrm{tr}\,A=\mathrm{tr}\,\frac{1}{2}\left(  ZAZ^{\dagger}+A\right)  $.
Therefore, without loss of generality, we suppose $A$ is diagonal. After some
elementary analysis, the optimal $A$ turns out to be
\[
A=\left\{
\begin{array}
[c]{cc}%
\mathrm{diag}\,\left(  a-\left\vert c\right\vert ,b-\left\vert c\right\vert
\right)  , & \left(  a\geq b\geq\left\vert c\right\vert \right)  \\
\mathrm{diag}\left(  a-\frac{\left\vert c\right\vert ^{2}}{b},0\right)  , &
\left(  a\geq\left\vert c\right\vert \geq b\right)
\end{array}
\right.
\]
and we have
\[
\mathrm{D}_{\left\vert 1-r\right\vert }^{\max}\left(  \rho\Vert\sigma\right)
=\left\{
\begin{array}
[c]{cc}%
4|c|=\left\Vert \rho-\sigma\right\Vert _{1}, & \left(  a\geq b\geq\left\vert
c\right\vert \right)  \\
2\left(  b+\frac{\left\vert c\right\vert ^{2}}{b}\right)  . & \left(
a\geq\left\vert c\right\vert \geq b\right)
\end{array}
\right.
\]

\end{document}